\documentclass[11pt]{article}
\usepackage{amssymb}
\usepackage{amsmath}
\usepackage{color}
\usepackage[colorlinks,linkcolor=blue,citecolor=green]{hyperref}
\newtheorem{theorem}{Theorem}

\newtheorem{example}[theorem]{Example}

\newtheorem{proposition}[theorem]{Proposition}
\newtheorem{remark}[theorem]{Remark}
\newenvironment{proof}[1][Proof]{\textbf{#1.} }{\ \rule{0.5em}{0.5em}}

\usepackage{graphicx}
\allowdisplaybreaks[3]
\begin{document}
\title{\Large \textbf{
Non-homogeneous Hamiltonian structures for quasilinear systems}}
\author{\large Pierandrea Vergallo$^{1,2}$ 
\\[3mm]
$^1$\small Department of Mathematical, Computer, Physical and Earth Sciences, \\
\small University of Messina, \\
\small V.le F. Stagno D'Alcontres 31, Messina, 98166, Italy\\
  $^2$\small Istituto Nazionale di  Fisica Nucleare -- Sez.\ Lecce
  \\
\small   \texttt{pierandrea.vergallo@unime.it}
}
\date{\small }

\maketitle
\abstract{This paper aims at investigating necessary (and sufficient) conditions for quasilinear systems of first order PDEs to be Hamiltonian, with non-homogeneous operators of order $1+0$, also with degenerate leading coefficient.  As a byproduct,  Tsarev's compatibility conditions are extended to degenerate operators.  Some examples are finally discussed.}

\maketitle
\section{Introduction}
Evolutionary  systems are very common in mathematical models for applied sciences, and play a central role in the theory of integrable systems. Indeed, they represent the natural extension to Partial Differential Equations (in the following, PDEs) of dynamical systems.  In general,  studying Nonlinear phenomena is a difficult task, especially in order to search for solutions of the systems.  Commonly, a powerful tool to investigate such models is represented by Hamiltonian formalism, as in the case of Ordinary Differential Equations (ODEs) in dynamical systems. In particular, finding a Hamiltonian structure for a chosen system gives arise to geometric considerations on the conserved quantities and the inner symmetries at first, but also on the solutions themselves at second.  Systems admitting the Hamiltonian property are largely present in physical and mechanical models, and are deeply studied in pure and applied mathematics.  Moreover, this formalism is also strictly connected to differential and algebraic geometry \cite{FerPavVit1,FPV14,VerVit2}, cohomological theories, algebraic topology and pure algebra. Several examples can be found in fluid dynamics, quantuum and soliton theories \cite{NovikovManakovPitaevskiiZakharov:TS} and recently also in kinetic models \cite{VerFer1}. Among these, the Korteweg de Vries equation, the Sine-Gordon equation, Euler's equations, the Camassa-Holm equations \cite{falq} and the KP equation are deeply investigated in the last forty years.  Interested reader can see \cite{DubKri,Olver:ApLGDEq, mokhov98:_sympl_poiss}. 

From a technical viewpoint,  let us consider evolutionary systems of order $k$:
\begin{equation}\label{1}
    u^i_t=F^i(x,u^j, u_x^j, \dots , u_{kx}^j), \qquad i=1,\dots, n.
\end{equation} 
in $1+1$ dimensions, i.e. the number of independent variables is two ($x,t$) and one of them has a favored role (usually  $t$ represents the time in models describing physical phenomena).
Here, the field variables are indicated with $u^i$,  for $i=1,2\dots , n$.  Now,  we briefly recall that an evolutionary system \eqref{1} admits Hamiltonian formulation if there exist a Hamiltonian differential operator $A^{ij}$ and a functional $H=\int{h(x,u^j,\dots u^j_{kx})\, dx}$ such that 
\begin{equation}
    u^i_t=F^i(x,u^j, u_x^j, \dots , u_{kx}^j)=A^{ij}\frac{\delta H}{\delta u^j}, \qquad i=1,2,\dots n.
\end{equation}
Hereafter, we consider local operators $A^{ij}=A^{ij\sigma}D_\sigma$, where $D_\sigma=D_x\circ \cdots \circ D_x$ is the composition of $\sigma$-times the total derivative with respect to $x$,  and indicate with $\delta/\delta u^j$ the variational derivative with respect to the field variable $u^j$.
The Hamiltonian property of the operator is given as follows. Consider a local differential operator and define a bracket between conserved quantities $P,Q$ as follows
\begin{equation}
    \{P,Q\}_A=\int{\frac{\delta p}{\delta u^i}A^{ij\sigma}D_\sigma\frac{\delta q}{\delta u^j}\, dx}\, .
\end{equation}
Then $A$ is Hamiltonian if the associated bracket $\{\, ,\, \}_A$ is a Poisson bracket, i.e. it is a derivation, skew-symmetric and satisfies the Jacobi identity. These conditions on the brackets are equivalent to require the operator $A$ to be skew-adjoint and that the Schouten brackets of $A$ with itself vanish ($[A,A]=0$).

\vspace{3mm}

The question naturally arising is: when does an evolutionary system admit Hamiltonian formulation? First results in this direction were presented in \cite{VerVit1} and later in \cite{Ver1}.  In this paper,  recent developments concerning quasilinear systems of first order PDEs are discussed.

\section{Quasilinear systems }
In the present paper, we focus on particular quasilinear systems of first order
\begin{equation}\label{a1}
u^i_t=V^i_j(u)u^j_x+W^i(u), \qquad i=1,2,\dots n.
\end{equation}
Systems \eqref{a1} naturally arise in models of mathematical biology or mathematical physics.  As an example,  consider the well-known $3$-wave equation:
\begin{equation}\begin{cases}\label{3wav}
u^1_t=-c_1u^1_x-2(c_2-c_3)u^2u^3\\
u^2_t=-c_2u^2_x-2(c_1-c_3)u^1u^3\\
u^3_t=-c_3u^3_x-2(c_2-c_1)u^1u^2\end{cases}
\end{equation}
Here $u^1,u^2,u^3$ are the field variables whereas $c_1,c_2,c_3$ are constants.

Moreover,  quasilinear systems arise when considering an evolutionary scalar equation of arbitrary order $k$. Indeed,  applying a procedure introduced by S.I.  Tsarev in \cite{tsa3},  it is possible to increase the number of field variables with the following change of coordinates:
\begin{equation}
u^1=u,  u^2=u_x,  \dots u^k=u_{(k-1)x}.
\end{equation}
Finally, inverting the independent variables $x$ with $t$ a quasilinear system \eqref{a1} of order 1 is obtained.  Interested readers can see \cite{tsa3} and the more recent paper \cite{DellVer1}. As an example,  using this inversion procedure the Kortweg De Vries equation $u_t=6uu_x-u_{xxx}$ is transformed into the following system
\begin{equation}
\begin{cases}u^1_t=u^2\\
u^2_t=u_3\\
u^3_t=-u^1_x+6u^2u^3
\end{cases}
\end{equation}

We deduce that the study of such evolutionary systems is of interest also in the theory of integrable systems.  
A first property of quasilinear systems of first order is given by the following Proposition.

\begin{proposition}Non-homogeneous quasilinear systems of first order \eqref{a1} are invariant in form under diffeomorfisms $\bar{u}^i=\bar{u}^i(u^j)$ of the manifold of dependent variables. \end{proposition}
\begin{proof}
Applying the standard rules for change of variables and derivatives,  the proof is straighforward.
\end{proof}
\subsection{Hamiltonian structures}

Let us briefly recall that homogeneous operators are operators of fixed degree $m$ in the order of derivation,  i.e.  operators for which every term is homogeneous of degree $m$ with respect to the $x$-derivative:
\begin{align}\label{hhom}
   \begin{split}
A^{ij}=&g^{ij}{\partial^m_x}+b^{ij}_k\,u^k_x\,{\partial^{m-1}_x}+\left(c^{ij}_k\,u^k_{xx}+c^{ij}_{kl}\,u^k_x\,u^l_x\right){\partial^{m-2}_x}+\dots \\
&\hphantom{ciaociao}\,\dots \,+\left(d^{ij}_k\,u^k_{nx}+\dots +d^{ij}_{k_1\dots k_m}\,u^{k_1}_x\,\cdots\, u^{k_m}_x\right) \,,
\end{split}
\end{align}  
where the coefficients $g^{ij},  b^{ij}_k,\, c^{ij}_k,\, \dots $ depend on the field variables.  Operators of this type were firstly introduced by Dubrovin and Novikov in \cite{DubrovinNovikov:PBHT} and later investigated also in \cite{FerPavVit1, VerVit1,VerVit2, FPV14,tsarev91:_hamil, mokhov98:_sympl_poiss} by different authors.  In particular,  in \cite{DN83} the authors introduced first order homogeneous Hamiltonian operators
\begin{equation}\label{dnop}
    g^{ij}\partial_x+b^{ij}_ku^k_x,
\end{equation}
where $g^{ij}$ and $b^{ij}_k$  depend on the field variables ${u^j}$.  Such operators are also known as \emph{Dubrovin-Novikov operators} and naturally arise when dealing with quasilinear systems of first order.  The Hamiltonianity conditions for \eqref{dnop} are given in the following Theorem (see the review \cite{mokhov98:_sympl_poiss}):
{\begin{theorem} \label{th:ham_A}
The operator \eqref{dnop} is Hamiltonian if and only if
\begin{subequations}
\begin{align} 
&g^{ij}=g^{ji} \\
&\frac{\partial g^{ij}}{\partial u^k} = b^{ij}_k+b^{ji}_k \\
&g^{is}b^{jk}_s-g^{js}b^{ik}_s=0 \\
&g^{is}\left( b^{jr}_{s,k}-b^{jr}_{k,s}\right)+b^{ij}_sb^{sr}_k-b^{ir}_s b^{sj}_k=0  \label{curv}\\
\begin{split}
&\left(g^{is}\left(b^{jr}_{s,k}- b^{jr}_{k,s}\right)+b^{ij}_sb^{sr}_k-b^{ir}_sb^{sj}_k\right)_{,q}+\displaystyle \sum_{(i j r)} \left(b^{si}_q\left(b^{jr}_{k,s}- b^{jr}_{s,k}\right)\right)\\
&\hphantom{ciao}+\left(g^{is}\left(b^{jr}_{s,q}- b^{jr}_{q,k}\right)+b^{ij}_sb^{sr}_q-b^{ir}_sb^{sj}_q\right)_{,k} + \displaystyle \sum_{(i j r)} \left(b^{si}_k\left( b^{jr}_{q,s}- b^{jr}_{s,q}\right)\right)=0 \,,
\end{split}
\end{align} \end{subequations}
with the sum over $(i,j,r)$ is on cyclic permutations of the indices. 
\end{theorem}}

In \cite{DN83}, Dubrovin and Novikov studied the non-degenerate case, i.e.  when the leading coefficient satisfies the condition $\det g^{ij}\neq 0$. Under this assumption,  the operator \eqref{dnop} is Hamiltonian if and only if $g_{ij}=(g^{ij})^{-1}$ is a flat metric and $b^{ij}_k=-g^{is}\Gamma_{sk}^j$, where $\Gamma^i_{jk}$ are Christoffel symbols for $g$.

A natural extension of homogeneous operators is given by non-homogeneous ones.  They are combinations of homogeneous operators of different orders.  In this paper, we will use the notation introduced by Dubrovin and Novikov \cite{DubrovinNovikov:PBHT}: if $C$ is a non-homogeneous operator obtained as  combination of two operators of degree $k$ and $h$ we say that $C$ is a non-homogeneous operator of type $k+h$ (with $k> h$). 

The simplest example is given by operators of type $1+0$,  also known as \emph{non-homogeneous hydrodynamic} operators and introduced in  \cite{DubrovinNovikov:PBHT}:
\begin{equation}
  C^{ij}=  g^{ij}\partial_x+b^{ij}_ku^k_x+\omega^{ij}
  \label{nhop}
\end{equation}
where $g^{ij},b^{ij}_k$ and $\omega^{ij}$ depend on the field variables ${u}^j$.

For the Hamiltonianity conditions of $C^{ij}$, we briefly recall the Hamiltonianity conditions of the $0$-order operator $\omega$:
\begin{theorem} \label{th:ham_omega}
The operator $\omega^{ij}(u)$ is Hamiltonian if and only if it forms  a finite-dimensional Poisson structure, i.e. it satisfies the conditions
\begin{equation}\label{cond3}
        \omega^{i j} = - \omega^{j i}  \qquad \textup{(skew-symmetry)}
\end{equation}
{\begin{equation} \label{cond4}
        \omega^{i s}  \omega^{j k}_{,s} + \omega^{j s} \omega^{k i}_{,s} + \omega^{k s} \omega^{i j}_{,s} =0  \,.
\end{equation}}\end{theorem} 

So that,  the following result holds.
\begin{theorem}\label{thm1}
The operator \eqref{nhop} is Hamiltonian if and only if $g^{i j}\, \partial_x + b^{i j}_{\,\, k} \, u^k_x$ 
is Hamiltonian, $\omega^{i j}$ is Hamiltonian, and the compatibility conditions are satisfied
{    \begin{align}\label{cond1}
    \Phi^{i j k} &= \Phi^{k i j} \,,
\\\label{eq:phi}
   \Phi^{i j k}_{,r}& = \sum_{(i,j,k)}  b^{s i}_{r}\, \omega^{j k}_{,s} + \left( b^{i j}_{r,s}- b^{i j}_{s,r}  \right)\omega^{s k}\,,  
\end{align}}
where $\Phi^{i j k}$ is the $(3,0)$-tensor
{\begin{align}
        \Phi^{i j k} & = g^{i s}\,  \omega^{j k}_{,s} - b^{i j}_{s} \, \omega^{s k} - b^{i k}_{s} \, \omega^{j s} \,. \end{align} }
\end{theorem}

Moreover,  in \cite{FerMok1} Ferapontov and Mokhov proved that 
\begin{proposition}If $\det g^{ij}\neq 0$,  the operator \eqref{nhop} is Hamiltonian if and only if 
\begin{equation}
\tilde{C}^{ij}=g^{ij}\partial_x+\Gamma^{ij}_ku^k_x
\end{equation}
is a Dubrovin-Novikov operator, 
$\omega^{ij}
$ defines a classical finite-dimensional Poisson structure and the following compatibilty conditions are satisfied:
\begin{gather}
\nabla^i\omega^{jk}+\nabla^{j}\omega^{ik}=0\label{condd1}\\
\nabla_s\nabla_k\omega^{ij}=0\label{condd2}
\end{gather} where $\nabla^i=g^{is}\nabla_s$ and $\nabla_s$ is the covariant derivative with respct to $g$.
\end{proposition}

As Ferapontov and Mokhov remarked in \cite{FerMok1},  condition \eqref{condd2} is just a corollary of \eqref{condd1}, then it does not give any additional requirement on $\omega$. 

Classifications of non-homogeneous operators of type $1+0$ were presented in both the degenerate and non-degenerate case, separately.
At first,  Dubrovin and Novikov proved \cite{DubrovinNovikov:PBHT} that if $\det g^{ij}\neq 0$, then there exist coordinates $\bar{u}^i=\bar{u}^i(u)$ such that the operator \eqref{nhop} takes the following form
\begin{equation}
C^{ij}=g^{ij}_0\partial_x+(c^{ij}_k\bar{u}^k+h^{ij}_0)
\end{equation}
where $g^{ij}_0=cost$ ($\Gamma^{ij}_k=0$), $c^{ij}_k$ are structural constants of the semisimple Lie algebra having $g^{ij}_0$ as Killing form and $h^{ij}_0$ is an arbitrary cocycle of the Lie algebra.  Otherwise, if $\det g^{ij}=0$ the complete classification of such operators was presented by Dell'Atti and the present author in \cite{DellVer1}.

\section{Compatibility conditions}\label{sec2}

The aim of this paper is to compute conditions which are necessary to find a Hamiltonian structure for a fixed quasilinear evolutionary system \eqref{a1}.  Hereafter, we will define such conditions as \emph{compatibility conditions}.  

\subsubsection{Differential coverings and compatibility conditions}

In order to compute compatibility conditions, we emphasize that in \cite{KerstenKrasilshchikVerbovetsky:HOpC}  the authors introduced a method to find necessary conditions for given systems to admit Hamiltonian formulation with a fixed operator.  The procedure makes use of the theory of differential coverings and was deeply investigated for homogeneous and nonlocal operators in \cite{VerVit1} and \cite{Ver1}.  In the following,  we will briefly show how it works.

Let us consider quasilinear systems of first-order partial differential equations
\begin{equation}\label{hyd}
{ F^i= u^i_t-V^i_j(u)u^j_x-W^i(u)=0\qquad i=1,2,\dots , n.}
\end{equation}

We recall that symmetries of \eqref{hyd} are vector functions $\varphi=\varphi^i$ such that $\ell_F(\varphi)=0$ when $F=0$, where $\ell_F$ is the Frech\'et derivative or the linearization of $\varphi$ with respect to $F$.  Explicitly:
\begin{equation}
\ell_F(\varphi)^i= D_t(\varphi^i)-(V^i_{j,l}u^j_x+W^i_{,l})\varphi^l-V^i_jD_x\varphi^j\qquad i=1,2,\dots , n.
\end{equation} Moreover, conservation laws are equivalence classes of 1-forms $\omega=a dt+b dx$ that are closed modulo $F=0$ up to total divergencies. A conservation law is uniquely represented by generating functions $\psi_j=\delta b/\delta u^j$, which are called cosymmetries of the system. Such vector functions $\psi$ satisfy $\ell_F^*(\psi)=0$, where $\ell_F^*$ is the formal adjoint of $\ell_F$, explicitly:
\begin{align}
\ell_F^*(\psi)_i&=-D_t\psi_i+(V^k_{i,j}u^j_x-V^k_{j,i}u^j_x-W^k_{,i})\psi_k+V^k_iD_x\psi_k\qquad i=1,2,\dots , n.
\end{align}

In \cite{Olver:ApLGDEq,KVV17},  it is shown that if $A$ is a Hamiltonian operator for $F=0$, then \begin{equation}\label{bivec}
\ell_F\circ A= A^*\circ \ell_F^*.
\end{equation}
where $A^*$ is the adjoint of the operator $A$. Then,  if $\psi_k$ is a cosymmetry it follows that $\ell_F(A^{ij}\psi_j)=0$,  i.e.  $\varphi^i=A^{ij}\psi_j$ is a symmetry of \eqref{hyd}.  Equivalently,  this means that Hamiltonian operators map conserved quantities into symmetries. 

In \cite{KerstenKrasilshchikVerbovetsky:HOpC}  the authors  introduce new variables $p_i$, such that it is possible to associate $D_x\psi_i$ to $p_{i,x}$, $D^2_{x}\psi_i$ to $p_{i,xx}$ and so on.  By this, they associate to each differential operator $A^{ij}=a^{ij\sigma}D_\sigma\psi_j$ a linear vector function $A^{i}=a^{ij\sigma}p_{j,\sigma}$ where $\sigma$ identifies the derivation order with respect to $x$.  For our purposes,  we introduce the \emph{cotangent covering} 
for systems \eqref{hyd}:
\begin{equation}\label{a22}
\mathcal{T}^*:\begin{cases}u^i_t=V^i_ju^j_x+W^i\\p_{i,t}=(V^k_{i,j}u^j_x-V^k_{j,i}u^j_x-W^k_{,i})p_k+V^k_ip_{k,x}\end{cases}
\end{equation}

Analogously,  we introduce new variables $q^i$ such that  we correspond $q^i$ for vector functions $\varphi^i$, $q^i_x$ to $D_x\varphi^i$ and so on. Then, we define the \emph{tangent covering} for non-homogeneous systems of type \eqref{hyd}:
\begin{equation}\label{aa1}
\mathcal{T}:\begin{cases}u^i_t=V^i_ju^j_x+W^i\\
q^i_t=(V^i_{j,l}u^j_x+W^i_{,l})q^l+V^i_jq^j_x\end{cases}.
\end{equation}

Both cotangent and tangent coverings are form-invariant under trasformations of type $\tilde{u}^i=\tilde{U}^i(\textbf{u})$. 
Finally,  the following result holds true:

\begin{theorem}[\cite{KerstenKrasilshchikVerbovetsky:HOpC}]A linear vector function $A$ in total derivatives satisfies \eqref{bivec} if and only if the equation
\begin{equation}\label{nes}
\ell_F(A(\textbf{p}))=0
\end{equation}holds on the cotangent covering \eqref{a22}.\end{theorem}

We can finally read conditions \eqref{nes} as necessary for system $F=0$ to be Hamiltonian with the structure given in terms of the operator $A^{ij}$.  So that, computing \eqref{nes} is equivalent to find  the compatibility conditions. 

The main advantadge of this approach is that it allows to compute explicit conditions without making use of a Hamiltonian functional $H=\int{h\, dx}$.  However, these \emph{compatibility conditions} do not guarantee that an operator $A$ is skew-adjoint or its Schouten brackets annihilate, then it does not ensure the Hamiltonianity of the operator.  For this reason,  they are only necessary.  

{ However, finding compatibilty conditions (even if only necessary \emph{a priori}) is a useful tool when studying the Hamiltonian formalism of a given evolutionary system. One can try to solve the obtained equations in order to search for a Hamiltonian structure or can exclude some operators if such conditions are not satisfied.  Moreover,  when the operator is non-degenerate,  one can invert the leading coefficient and search for an explicit Hamiltonian density. }

\vspace{5mm}

{ In what follows,  we consider operators $A$ to be Hamiltonian. } Moreover, we will focus on quasilinear systems only.  To this aim we gradually analyse 0-order systems,  first order homogeneous systems and finally first order non-homogeneous ones.

\vspace{3mm}

\subsubsection*{A.  0-order systems}

Let us firstly consider 0-order systems and search for a Hamiltonian formulation with a 0-order operator.
\begin{proposition}\label{lem1}Let us consider the system
\begin{equation}\label{a2}
u^i_t=W^i(u), \qquad i=1,2,\dots n
\end{equation}

Then,  \eqref{a2} admits Hamiltonian formulation with the non-degenerate ultralocal operator $\omega^{ij}$ if and only if 
\begin{equation}\label{890}\tilde{\nabla}^iW^j=\tilde{\nabla}^jW^i,\end{equation}
where {$\tilde{\nabla}^i=\omega^{is}\tilde{\nabla}_s$} and {$\tilde{\nabla}_s$ is the covariant derivative with respect to $\omega$, where
\begin{equation}
\tilde{\Gamma}^j_{ik}=-\frac{1}{2}\omega_{ia}\omega^{aj}_{,k}.
\end{equation} are the related symbols of the connection.}
\end{proposition}
\begin{proof}[Proof of \ref{lem1}] Let us assume that $\omega$ is non-degenerate.  Then, the system is Hamiltonian if $W^i=\omega^{is}{\nabla}_sh$ for a certain  function $h=h(u)$. Such condition is satisfied if and only if the system is compatible, that is {$h_{,js}=h_{,sj}$}.  Having 
{$h_{,j}=\omega_{js}W^s$},  the compatibility conditions become:
\begin{equation}
\omega_{si,j}W^i+\omega_{si}W^i_{,j}-\omega_{ji,s}W^i-\omega_{ji}W^i_{,s}=0
\end{equation} by applying twice the inverse 2-form $\omega^{ij}$ we obtain $\tilde{\nabla}^iW^j-\tilde{\nabla}^jW^i=0$. \end{proof}

\vspace{3mm}

\begin{remark}Analogously we can compute $\ell_F(\omega^{ij}p_j)=0$, following the procedure described at the beginning of this section.  In this case,  the compatibility condition is given by 
{\begin{equation}\label{8901}
\omega^{ij}_{,s}W^s-\omega^{js}W^i_{,s}-\omega^{is}W^j_{,s}=0.
\end{equation} } \eqref{8901}  is obtained  without the assumption of non-degeneracy of $\omega$.  Note that  it reduces to \eqref{890} when $\omega$ is invertible. \end{remark} 

\vspace{5mm}

\subsubsection*{B.  First order homogeneous systems}
Let us now focus on the homogeneous case of quasilinear systems of first order,  i.e.  on systems of type 
\begin{equation}\label{homo}
u^i_t=V^i_j(u)u^j_x.
\end{equation} A natural assumption in Hamiltonian theories is to require the Hamiltonian density $H$ to depend only on the field variables $u^j$ (and not on their derivatives).  In such a case,  we usually say the density $h=h(u)$ to be of \emph{hydrodynamic type}.  Then,  system \eqref{homo} is Hamiltonian in the hydrodynamic sense if 
\begin{equation}
V^i_j=\nabla^i\nabla_jh=g^{is}\nabla_s\nabla_jh
\end{equation}where $g^{ij}$ is a flat non-degenerate metric and $\nabla$ is the Riemannian covariant derivative of $g$.  Note that systems \eqref{homo} are also known as hydrodynamic type systems.  Moreover,  following a notation introduced by Tsarev,  we define the matrix $V^i_j$ to be a \emph{Hamiltonian matrix}.

Let us now compute conditions similar to Proposition \ref{lem1} for hydrodynamic type systems and first order homogeneous Hamiltonian operators (not necessarily in the non-degenerate case).  In order to do this we apply the above described method.  Then,  let us consider a first-order homogeneous {\color{red} operator}:
\begin{equation}\label{eqse}
A^{ij}=g^{ij}\partial_x+b^{ij}_ku^k_x
\end{equation}
We now associate to \eqref{eqse} the linear vector function $A^i(\textbf{p})=g^{ij}p_{j,x}+b^{ij}_ku^k_xp_j$.

{\begin{theorem}\label{thm8}
If a quasilinear homogeneous system admits a Hamiltonian formulation with a Hamiltonian homogeneous operator of first order, the following conditions are then satisfied: 
 \begin{align}
  \label{eq:100}
  & g^{is}V^j_s=g^{js}V^i_s,
    \\   
    &\label{eq:103}
      g^{is}\left(V^j_{s,k}-V^j_{k,s}\right)+
      b^{ij}_sV^s_k-b^{sj}_kV^i_s = 0.
  \end{align}
\end{theorem}}
In the specific case of non-degenerate $g^{ij}$, we substitute  coefficients $b^{ij}_k$ with $\Gamma^{ij}_k$, where {$b^{ij}_k=-g^{is}b_{sk}^j$} and the latter are exactly Christoffel symbols of second kind.  Then,  Tsarev proved the following theorem:
\begin{proposition}[\cite{tsarev85:_poiss_hamil,tsarev91:_hamil}]\label{tsa}
A quasilinear system is Hamiltonian in Dubrovin-Novikov sense if and only if there exists a nondegenerate flat metric $g_{ij}$ such that 
\begin{align}
g_{is}V^{s}_j=g_{js}V^s_i,\\
\nabla_iV^j_k=\nabla_kV^j_i,
\end{align}
where $g^{ij}=(g_{ij})^{-1}$.
\end{proposition}
As natural,  the previous can be considered a Corollary of Theorem \ref{thm8}. 

\vspace{3mm}

\begin{proof}[Proof of \ref{thm8}]
{ Let us assume the operator $A^{ij}=g^{ij}\partial_x+b^{ij}_ku^k_x$ to be Hamiltonian, i.e. conditions in Theorem \ref{th:ham_A} are satisfied.  In order to obtain the necessary conditions, we only need to compute the compatibility conditions $\ell_F(A^i(\textbf{p}))=0$, for $i=1,2,\dots n$.}

Now, 
{\begin{equation}
  \label{eq:23}
  \begin{split}
  \ell_F(A(\mathbf{p})) =& \partial_t(g^{ij}p_{j,x} +
 b^{ij}_k u^k_x p_j)
  \\
  & - V^i_{l,k}u^l_x(g^{kj}p_{j,x} + b^{kj}_h u^h_x p_j)
   - V^i_k\partial_x(g^{kj}p_{j,x} + b^{kj}_h u^h_x p_j)
\end{split}
\end{equation}}
Then, expliciting the identities on the cotangent covering and collecting for $u^j_\sigma$, we obtain the following set of conditions
 \begin{align}
  \label{eq:100}
  & V^i_kg^{kj}-V^j_kg^{ki}= 0,
    \\
    \label{eq:101}
    \begin{split}
      & g^{ij}_kV^k_m+g^{ik}(V^j_{k,m}-V^j_{m,k})+g^{ik}V^j_{k,m}
      +b^{ik}_mV^j_k
      -V^i_{m,k}g^{kj}-V^i_kg^{kj}_m-V^i_kb^{kj}_m=0,
    \end{split}
    \\
    &\label{eq:103}
      g^{ik}\left(V^j_{k,h}-V^j_{h,k}\right)+
      b^{ij}_kV^k_h-b^{kj}_hV^i_k = 0,
    \\
    \label{eq:201}
    \begin{split}
      &g^{ik}\left(V^j_{k,ml}+V^j_{k,lm}-V^j_{m,kl}-V^j_{l,km}\right)
      +b^{ij}_{m,k}V^k_l+b^{ij}_{l,k}V^k_m+b^{ij}_kV^k_{l,m}
      +b^{ij}_kV^k_{m,l}
      \\
      &\hphantom{ciao}
      +b^{ik}_lV^j_{k,m}+b^{ik}_mV^j_{k,l}-b^{ik}_lV^j_{m,k}
      -b^{ik}_mV^j_{l,k}\\
      &\hphantom{ciao}
      -b^{kj}_mV^i_{l,k}-b^{kj}_lV^i_{m,k}-b^{kj}_{m,l}V^i_k
      -b^{kj}_{l,m}V^i_k=0.
    \end{split}
  \end{align}
  
 Let us first notice that condition \eqref{eq:101} is a differential consequence of \eqref{eq:100} and \eqref{eq:103}.  Indeed,  using the Hamiltonianity condition $g^{ij}_{,k}=b^{ij}_k+b^{ji}_k$ of the operator, we exactly obtain that 
 \begin{multline}
g^{ik}(V^j_{k,h}-V^j_{h,k})+b^{ij}_kV^k_h-b^{kj}_hV^i_k
g^{jk}(V^i_{k,h}-V^i_{h,k})\\ +b^{ji}_kV^k_h-b^{ki}_hV^j_k+\partial_h(g^{ik}V^j_k-g^{jk}V^i_k)=0,
\end{multline}
that is trivially equal to zero when \eqref{eq:100} and \eqref{eq:103} are taken into account. Note that this proof is exactly the same as the one of Lemma 4 in \cite{VerVit1}.  On the other side,  we can again prove that \eqref{eq:201} is a differential consequence of \eqref{eq:100} and \eqref{eq:103} but the proof needs some arrangements from \cite{VerVit1} due to the fact that the tensor $g^{ij}$ is not invertible in this case. 

 Let us subtract the differential consequence 
{ \begin{align}
 \left(g^{ik}(V^j_{k,m}-V^j_{m,k})
        +b^{ij}_kV^k_m-b^{kj}_mV^i_k\right)_{,l}+\left(g^{ik}(V^j_{k,l}-V^j_{l,k})
        +b^{ij}_kV^k_l-b^{kj}_lV^i_k\right)_{,m}\end{align}}
        of \eqref{eq:103} from equation
  \eqref{eq:201}, obtaining
\begin{align*}
     &\left(b^{ij}_{m,k}-b^{ij}_{k,m}\right)V^k_l
     +\left(b^{ij}_{l,k}-b^{ij}_{k,l}\right)V^k_m+ b^{kj}_l\left(V^i_{k,m}-V^i_{m,k}\right)+
        b^{ki}_l\left(V^j_{m,k}-V^j_{k,m}\right)
        \\
       & + b^{kj}_m\left(V^i_{k,l}-V^i_{l,k}\right)+
     b^{ki}_m\left(V^j_{l,k}-V^i_{k,l}\right)
\end{align*}
Now,  let us multiply for $g^{am}g^{bl}$, obtaining
\begin{align}\begin{split}
&g^{am}\left(b^{ij}_{m,k}-b^{ij}_{k,m}\right)g^{bl}V^k_l
     +g^{bl}\left(b^{ij}_{l,k}-b^{ij}_{k,l}\right)g^{am}V^k_m+ g^{bl}b^{kj}_lg^{am}\left(V^i_{k,m}-V^i_{m,k}\right)
        \\
       & +
     g^{bl}  b^{ki}_l g^{am}\left(V^j_{m,k}-V^j_{k,m}\right)+ g^{am}b^{kj}_mg^{bl}\left(V^i_{k,l}-V^i_{l,k}\right)+
     g^{am}b^{ki}_mg^{bl}\left(V^j_{l,k}-V^i_{k,l}\right)\end{split}
\end{align}
Let us  use condition \eqref{eq:101} for the last four terms.  Then, by using \eqref{eq:100}, some terms erase and we obtain the following
\begin{align}\begin{split}\label{quant}
&g^{bl}V^k_l\left(g^{am}\left(b^{ij}_{m,k}-b^{ij}_{k,m}\right)+b^{ai}_mb^{mj}_k-b^{aj}_mb^{mi}_k\right)\\
&\hphantom{ciao}+g^{al}V^k_l\left(g^{bm}\left(b^{ij}_{m,k}-b^{ij}_{k,m}\right)+b^{bi}_mb^{mj}_k-b^{bj}_mb^{mi}_k\right)
\end{split}
\end{align}
  We finally obtain that, by condition \eqref{curv},  \eqref{quant} is zero and the Theorem is proved.
\end{proof}

\vspace{3mm}

As a corollary, in the non-degeneracy assumption on $g^{ij}$, we prove Corollary \ref{tsa}:

\vspace{3mm}
\begin{proof}[Proof of \ref{tsa}]
The first condition is a direct consequence of the first one in Theorem \ref{thm8}, by lowering the indices.  For the second condition, it is sufficient to observe that 
\begin{equation}
g^{is}\left(\nabla_sV^j_k-\nabla_kV^j_s\right)=g^{is}\left(V^j_{k,s}-V^j_{s,k}\right)+\Gamma^{ij}_sV^s_k+g^{is}V_s^l\Gamma^j_{kl}
\end{equation}
and by using $g^{is}V_s^l=g^{ls}V_s^i$ we obtain 
\begin{equation}
g^{is}\left(\nabla_sV^j_k-\nabla_kV^j_s\right)=g^{is}\left(V^j_{k,s}-V^j_{s,k}\right)+\Gamma^{ij}_sV^s_k-\Gamma^{sj}_kV^i_s.
\end{equation}
Due to the non-degeneracy of $g^{ij}$, we obtain that condition 2 in Theorem \ref{thm8} is satisfied if and only if Tsarev's conditions are satisfied.
\end{proof}

\vspace{5mm}

\subsubsection*{C.  First order non-homogeneous systems}

Finally,  we extend the previous procedures to non-homogeneous hydrodynamic-type systems and non-homogeneous operators of type $1+0$. 
It follows that if a non-homogeneous hydrodynamic-type system is Hamiltonian the following holds
\begin{equation}
u^i_t={V^i_ju^j_x+W^i=(u^j_x\nabla^i\nabla_j+\tilde{\nabla}^i)h}\qquad i=1,2,\dots , n.
\end{equation}
{Here we denote with $\nabla_i$ the covariant derivative with respect to the metric tensor $g$, indicating the related Christoffel symbols with $\Gamma^i_{jk}$, whereas we use $\sim$ to indicate the covariant derivative with respect to the symplectic form $\omega$, choosing to indicate the related symbols with $\tilde{\Gamma}^i_{jk}$ and the derivative with $\tilde{\nabla}$.}

Let us now consider a non-homogeneous operator of type 1+0 satisfying Theorem \ref{thm1}:
\begin{equation}
C^{ij}=A^{ij}+\omega^{ij}=g^{ij}\partial_x+b^{ij}_ku^k_x+\omega^{ij},
\end{equation}
that can be identified by the linear vector function
$
C^{ij}(\textbf{p})=g^{ij}p_{j,x}+b^{ij}_ku^kp_j+\omega^{ij}p_j.
$
The following theorem holds.
{\begin{theorem}\label{thmcomp} If a quasilinear system
\begin{equation}\label{sys}
u^i_t=V^i_j(u)u^j_x+W^i(u), \qquad i=1,2,\dots n
\end{equation}is Hamiltonian with the Hamiltonian operator $C^{ij}$, then
the following conditions are satisfied:
\begin{align}
&V^i_sg^{sj}-V^j_sg^{si}= 0;\\
&g^{is}\left(V^j_{s,k}-V^j_{k,s}\right)+
      b^{ij}_sV^s_k-b^{sj}_kV^i_s = 0;\\
&{g^{ij}_{,s}W^s-g^{js}W^i_{,s}-g^{is}W^j_{,s}-\omega^{si}V^j_s-\omega^{sj}V^i_s=0;}\\
&W^i_{,s}\omega^{sj}+W^j_{,s}\omega^{is}+\omega^{ji}_{,s}W^s=0;\\
&T^{ij}_k=0 \end{align}
where
\begin{align*}T^{ij}_k&=-g^{il}W^j_{,lk}+b^{ij}_{k,l}W^l+b^{ij}_lW^l_{,k}-b^{il}_kW^j_{,l}-b^{lj}_kW^i_{,l}-V^i_{k,s}\omega^{sj}\\
&\hphantom{ciaociao}+\omega^{is}\left(V^j_{s,k}-V^j_{k,s}\right)+\omega^{ij}_{,s}V^s_k-\omega^{sj}_{,k}V^i_s\end{align*}
\end{theorem}}

Whereas,  under the non-degeneracy hypothesis on $g^{ij}$:

{\begin{proposition}\label{cor2}
If the quasilinear system
\begin{equation}\label{sys}
u^i_t=V^i_j(u)u^j_x+W^i(u), \qquad i=1,2,\dots n
\end{equation}is Hamiltonian with the non-degenerate Hamiltonian operator $C^{ij}$, then
the following conditions are satisfied:
\begin{align}
&\nabla^iV^j_k=\nabla^jV^i_k;\\
&g^{ik}V^j_k=g^{jk}V^i_k;\\
&\nabla^iW^j+\nabla^jW^i=\omega^{ik}V^j_k+\omega^{jk}V^i_k;\\
&\tilde{\nabla}^iW^j=\tilde{\nabla}^jW^i;\label{527}\\
&T^{ij}_k=0.
\end{align}
\end{proposition}}
\begin{proof}[Proof of \ref{thmcomp}] {Let us assume that $C^{ij}=g^{ij}\partial_x+b^{ij}_ku^k_x+\omega^{ij}$ is Hamiltonian (see conditions in Theorem \ref{thm1}).  We only need to prove the compatibility conditions $\ell_F(C^i(\textbf{p}))=0$ for $i=1,2,\dots n$. }

Now,  let us observe that 
\begin{equation}
\ell_F(C)=\ell_F(A)+\ell_F(\omega^{ij}p_j)
\end{equation}
Then 
{\begin{align}\begin{split}
\ell_F(\omega^{ij}p_j)&=\partial_t(\omega^{ij}p_j)-\left(V^i_{j,l}u^j_x+W^i_l\right)\omega^{lk}p_k-V^i_j\partial_x(\omega^{jk}p_k)\\
&=\omega^{ij}_{,k}u^k_tp_j+\omega^{ij}p_{j,t}-V^i_{j,l}\omega^{lk}u^j_xp_k-W^i_{,j}\omega^{lk}p_k+\\
&\hphantom{cioaciaociaociao}-V^i_j\omega^{jk}_{,l}u^l_xp_k-V^i_j\omega^{jk}p_{k,x}\\
&=\omega^{ij}_{,k}\left(V^k_lu^l_x+W^k\right)p_j+\omega^{ij}\left(V^k_{j,l}u^l_x-V^k_{l,j}u^l_x-W^k_{,j}\right)p_k\\
&\hphantom{==ciao}+\omega^{ij}V^k_jp_{k,x}-V^i_{j,l}\omega^{lk}u^j_xp_k-W^i_{,l}\omega^{lk}p_k+\\
&\hphantom{cioaciaociaociao}-V^i_j\omega^{jk}_{,l}u^l_xp_k-V^i_j\omega^{jk}p_{k,x}\end{split}
\end{align}}
Summing this with the expression of $\ell_F(A)$ (see \cite{Ver1}) and collecting for each variable $u^j_\sigma$, the {theorem} is proved.
\end{proof}

\vspace{3mm}
\begin{proof}[Proof of \ref{cor2}]
Then,  with non-degenerate coefficients $g^{ij}$ and $\omega^{ij}$ we  introduce the following symbols:
\begin{equation}
\tilde{\Gamma}^j_{ik}=-\frac{1}{2}\omega_{ia}\omega^{aj}_{,k}
\end{equation}

this is a symplectic connection which is compatible with the symplectic form $\omega^{ij}$ (see also \cite{mokhov98:_sympl_poiss}). Define $\tilde{\Gamma}^{ij}_k=\omega^{is}\tilde{\Gamma}^j_{sk}$, then condition \eqref{527} in the previous theorem can be written in a coordinate-free form as follows:
\begin{equation}
\tilde{\nabla}^jW^i-\tilde{\nabla}^iW^j=0
\end{equation}
where $\tilde{\nabla}$ is the covariant derivative for the connection $\tilde{\Gamma}^{i}_{jk}$.

\end{proof}

\section{Examples}

This section aims at presenting some examples of quasilinear systems satisfying the compatibility conditions found in Section \ref{sec2}.  The examples treated herein can be computed via Computer Algebra Systems  as Reduce (see \cite{KVV17,Vit1}) or Maple.

Let us firstly study the case of KdV equation, in two versions. Recall that in both cases we applied the method of inversion of the system in order to increase the number of variables and reduce the order (see \cite{tsa3,DellVer1})

\begin{example}[KdV equation - I]
Let us consider the KdV equation:
\begin{equation}\label{kdv}
u_t=6uu_x-u_{xxx}
\end{equation}
and let us regard \eqref{kdv} as an evolutionary system:
\begin{equation}\label{kdvsys}
\begin{cases}
u^1_t=u^2\\u^2_t=u^3\\u^3_t=-u^1_x+6u^1u^2
\end{cases}
\end{equation}
Then it is a non-homogeneous hydrodynamic type system. Let us consider the non-homogeneous operator
\begin{equation}\label{kdvop}
C^{ij}=g^{ij}\partial_x+\Gamma^{ij}_ku^k_x+\gamma^{ij}
\end{equation}
where $$g^{ij}=\begin{pmatrix}
0&0&1\\0&-1&0\\1&0&8u^1
\end{pmatrix}$$
and $$\gamma^{ij}=\begin{pmatrix}
0&2u^1&2u^2\\-2u^1&0&-12(u^1)^2+2u^3\\-2u^2&12(u^1)^2-2u^3&0
\end{pmatrix}$$
The compatibility conditions described in Corollary \ref{cor2} are satisfied.

\end{example}
\begin{example}[KdV equation - II]\label{ex0}
Let us consider again system \eqref{kdvsys}.  
This system is also Hamiltonian with the following non-homogeneous hydrodynamic type operator \cite{tsa3}
\begin{equation}
    C^{ij}=\begin{pmatrix}
    0&0&0\\0&0&0\\0&0&-1
    \end{pmatrix}\partial_x+\begin{pmatrix}
    0&-1&0\\1&0&6u^1\\0&-6u^1&0
    \end{pmatrix} \,,
    \end{equation}
   with the leading coefficient $g^{ij}$ being degenerate.    The operator and the system are compatible in the sense of Theorem \ref{thmcomp}.
\end{example}
Note that in Example \ref{ex0}, the Hamiltonian operator is degenerate.  Other examples in 2 components of non-homogeneous operators with degenerate leading coefficient are given in the following:

\begin{example}[$2$-wave interaction system]
In \cite{mokhov98:_sympl_poiss}, the author introduced the real reduction of 2-waves interaction system formulated in terms of the system of hydrodynamic equations: 
\begin{equation}\label{mok1}
\begin{cases}
u_t=a uv \\
v_t=a v_x+u^2
\end{cases}\,,
\end{equation}
with $a$ constant. The system admits a Hamiltonian formulation, with the operator
\begin{equation}
C^{i j}=\begin{pmatrix}
0&0\\0&1
\end{pmatrix}\partial_x\,+\,\begin{pmatrix}
0&-u\\u&0
\end{pmatrix} \,.
\label{eq:ham_op_2waves}
\end{equation} 
As expected, the system is compatible in the sense of Theorem \ref{thmcomp} with the operator $C^{ij}$. 
\end{example}

\begin{example}[Sinh-Gordon equation]

{Let us consider the Sinh-Gordon equation
\begin{equation}
    \varphi_{\tau\xi}=\sinh{\varphi} \,.
\end{equation}
Applying the change of variables given by $\varphi=2\log u$ and  $v={2u_\tau}/{u}$,  consider the light-cone coordinates $\tau=t, \xi=t-x$ and obtain
\begin{equation}
    \begin{cases}u_t=\dfrac{1}{2}uv\\
    v_t=v_x+\dfrac{1}{2}\left(u^2-\dfrac{1}{u^2}\right)
    \end{cases}.
\end{equation}
In \cite{DellVer1}, the system is shown to be Hamiltonian with the non-homogeneous hydrodynamic operator
\begin{equation}
    C^{i j}=\begin{pmatrix}
0&0\\0&1
\end{pmatrix}\partial_x\,+\frac{1}{2}\begin{pmatrix}
0&u\\-u&0
\end{pmatrix} .
\end{equation}
The operator and the system are compatible in the sense of Theorem \ref{thmcomp}.}
\end{example}

Finally, we conclude this section with the first example presented in the Introduction.

\begin{example}[The 3-waves equation]
Let us consider the 3-waves equations:
\begin{equation}\begin{cases}\label{3wav}
u^1_t=-c_1u^1_x-2(c_2-c_3)u^2u^3\\
u^2_t=-c_2u^2_x-2(c_1-c_3)u^1u^3\\
u^3_t=-c_3u^3_x-2(c_2-c_1)u^1u^2\end{cases}
\end{equation}
and let us consider the non-homogeneous local  operator
\begin{equation}
M^{ij}=\begin{pmatrix}
1&0&0\\0&-1&0\\0&0&-1
\end{pmatrix}\partial_x+\begin{pmatrix}
0&-2u^3&2u^2\\
2u^3&0&2u^1\\-2u^2&-2u^1&0
\end{pmatrix}.
\end{equation}
The compatibility conditions presented in Corollary \ref{cor2} are satisfied. 
\end{example}

An analogue of this discussion, but in terms of nonlocal operators for non-homogeneous quasilinear systems, was also presented in \cite{Ver1}.  In the paper,  the author showed compatibility conditions for homogeneous first order operators similar to Tsarev's ones and extended to non-homogeneous systems.  Other examples are also therein presented.

\section{Conclusions}
In this paper,  tensorial  necessary conditions for a given quasilinear system of first order PDEs to admit a Hamiltonian structure are discussed.  We emphazise that such conditions were firstly presented for non-degenerate homogeneous Hamiltonian operators of order 1 and homogeneous systems (see \cite{tsarev85:_poiss_hamil}).  

The main novelties of the present work are the following:
\begin{itemize}
\item degenerate leading coefficients of the operators are discussed,  
\item conditions previously computed are extended to non-homogeneous operators of type $1+0$, and 
\item non-homogeneous systems are investigated.  
\end{itemize}
Some example of equations of the previous kind are presented (such as the well-known Korteweg-de Vries treated as system,  the 3-waves equation,  and the 2-wave interaction system),  showing that a further investigation of such conditions for different operators and systems is needed in the future.

\vspace{3mm}

\textbf{Acknowledgements.} The author thanks R. Vitolo,  M. Menale and M.  Dell'Atti for stimulating discussions. The author also acknowledges the financial support of GNFM of the Istituto Nazionale di Alta Matematica, of PRIN 2017 \textquotedblleft Multiscale phenomena in Continuum
Mechanics: singular limits, off-equilibrium and transitions\textquotedblright,
project number 2017YBKNCE and of the research project Mathematical Methods in Non-Linear Physics
(MMNLP) by the Commissione Scientifica Nazionale -- Gruppo 4 -- Fisica Teorica
of the Istituto Nazionale di Fisica Nucleare (INFN).

\vspace{5mm}

\textbf{Conflict of interests.}
On behalf of all authors, the corresponding author states that there is no conflict of interest.

\end{document}